\documentclass{style/notices}
\usepackage{amsfonts,amssymb,amsmath,amscd,amsthm}
\usepackage{graphicx}
\usepackage{hyperref}
\hypersetup{pdfborder = 0 0 0}
\usepackage{cleveref}
\usepackage{libertinus}
\usepackage[colorinlistoftodos]{todonotes}
\usepackage{float}

\title{
    Measuring Diversity and Segregation with Convex Functions
}

\author{
  Philip S. Chodrow
  \affil{
    Phil is an assistant professor of computer science at Middlebury College. His email address is \tt{pchodrow@middlebury.edu}.
    }
}

\definecolor{lightgreen}{rgb}{0.678, 0.847, 0.678}

\theoremstyle{definition}
\newtheorem{dfn}{Definition}

\theoremstyle{plain}
\newtheorem{thm}{Theorem}[section]

\crefname{dfn}{Definition}{Definitions}
\crefname{thm}{Theorem}{Theorems}
\crefname{prop}{Proposition}{Propositions}

\theoremstyle{remark}
\newtheorem{example}{Example}[section]
\crefname{example}{Example}{Examples}

\crefname{remark}{Remark}{Remarks}

\newcommand{\N}{\mathbb{N}}
\newcommand{\R}{\mathbb{R}}

\newcommand{\simplex}[1]{\Delta_{#1}}
\newcommand{\paren}[1]{\left( #1 \right)}
\newcommand{\braces}[1]{\left\{ #1 \right\}}
\newcommand{\brackets}[1]{\left[ #1 \right]}
\newcommand{\vy}{\mathbf{y}}
\newcommand{\mX}{\mathbf{X}}
\newcommand{\mY}{\mathbf{Y}}
\newcommand{\mV}{\mathbf{V}}
\newcommand{\mU}{\mathbf{U}}
\newcommand{\cX}{\mathcal{X}}

\newcommand{\abs}[1]{\left| #1 \right|}
\newcommand{\norm}[1]{\lVert #1 \rVert}

\newcommand{\exc}[5]{\mathrm{exc}_{#1,#2, #3\leftrightarrow#4}^{(#5)}}
\DeclareMathOperator*{\argmax}{argmax}

\newcommand{\ve}{\mathbf{e}}

\newcommand{\vx}{\mathbf{x}}

\newcommand{\trace}[1]{\mathrm{trace}\;#1}

\newcommand{\vdelta}{\boldsymbol{\delta}}

\newcommand{\Pops}{\mathbb{P}}
\newcommand{\pop}{\mathcal{P}}
\newcommand{\cS}{\mathcal{S}}
\newcommand{\info}{I}

\newcommand{\divergence}[3]{D_{#1}\paren{#2 \;\|\; #3}}

\newcommand{\hess}[1]{\nabla^2 #1}

\newcommand{\vmu}{\boldsymbol{\mu}}

\newcommand{\cR}{\mathcal{R}}
\newcommand{\jacobian}[1]{\nabla #1}
\newcommand{\relint}[1]{#1^{\circ}}

\newcommand{\mC}{\mathbf{C}}

\newcommand{\standardTriple}{(\mX, \mY, \vmu)}

\newcommand{\inner}[2]{#1^\top #2}

\newcommand{\eqdef}{\triangleq}

\newcommand{\vyref}{\mathbf{y}_{\mathrm{ref}}}

\newcommand{\smooth}{\kappa}

\begin{document}
\maketitle

\section{Introduction}

We consider the problem of \emph{segregation measurement}. 
Colloquially, a population of individuals is segregated with respect to a given trait if individuals with different values of that trait tend to be separated from one another by geography, policy, or other social forces.
The segregation measurement problem is to construct compact, quantitative descriptions of the extent and structure of segregation in a given population. 
In the quantitative sociology literature, a prototypical segregation measurement is an ``index:'' a single number intended to describe the extent and structure of segregation in a population.
My focus here will be on the common case of \emph{spatial} segregation, in which individuals are assumed to be located in space (often at their place of residence) and we seek to understand how strongly space separates individuals with different demographic traits. 

To motivate this problem, as well as some of its complexities, \Cref{fig:dot-viz} visualizes the spatial distribution of racial demographics in Detroit, Michigan and Atlanta, Georgia.\footnote{
  I obtained these data for this article from the 2024 American Community Survey (ACS) 5-year estimates, which are available from the U.S. Census Bureau at \url{https://data.census.gov/}.
  I then aggregated the resulting table, which contains a total of 16 distinct ethnic and racial categories, into the five supercategories shown in each visualization.
  }
Of these two cities, we may wish to measure how \emph{diverse} or \emph{segregated} each one is, and perhaps compare them. 
To do so, we must decide how the \emph{spatial} distribution of individuals should factor into our analyses. 
We must also grapple with the \emph{modifiable areal unit problem} (MAUP) \cite{openshaw1984modifiable}: our spatial information is typically mediated by the geographic units in which our data are reported, such as counties, county subdivisions (pictured here), census tracts, or block groups; changing these boundaries would alter our segregation measurements.
Since these units are often arbitrary from an analytical point of view, we must also aim to construct measures which are robust to the choice of spatial partitioning.

\begin{figure}
  \centering
  \includegraphics[width=1\linewidth]{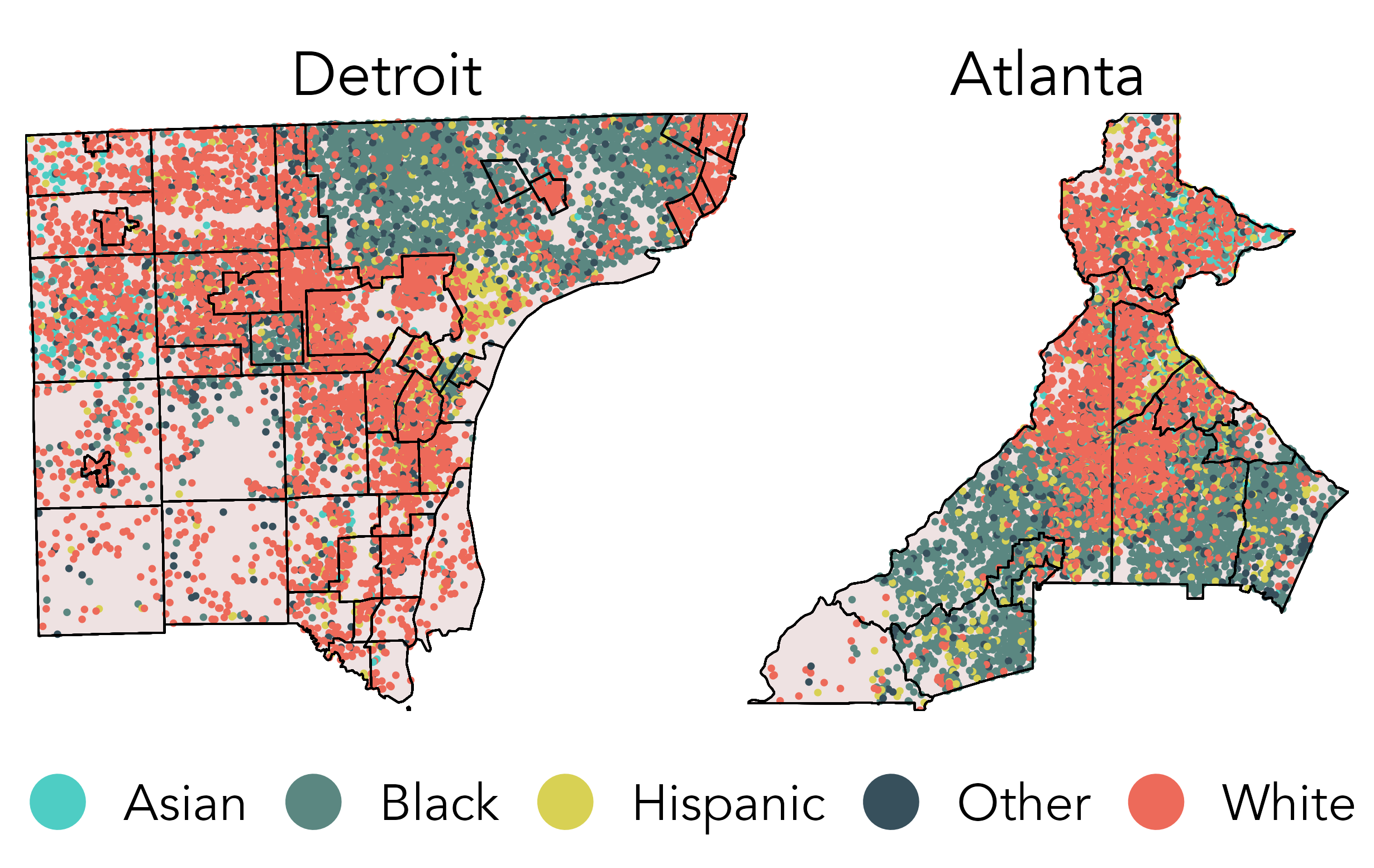}
  \caption{
    Dot visualization of race/ethnicity demographics in Detroit, Michigan and Atlanta, Georgia.
    There are 10,000 dots drawn, each of which represents an equal unit of population from the specified demographic category. 
    County subdivision boundaries are drawn in black.
  }\label{fig:dot-viz}
\end{figure}

\begin{figure*}[h]
  \centering
  \includegraphics[width=1.0\textwidth]{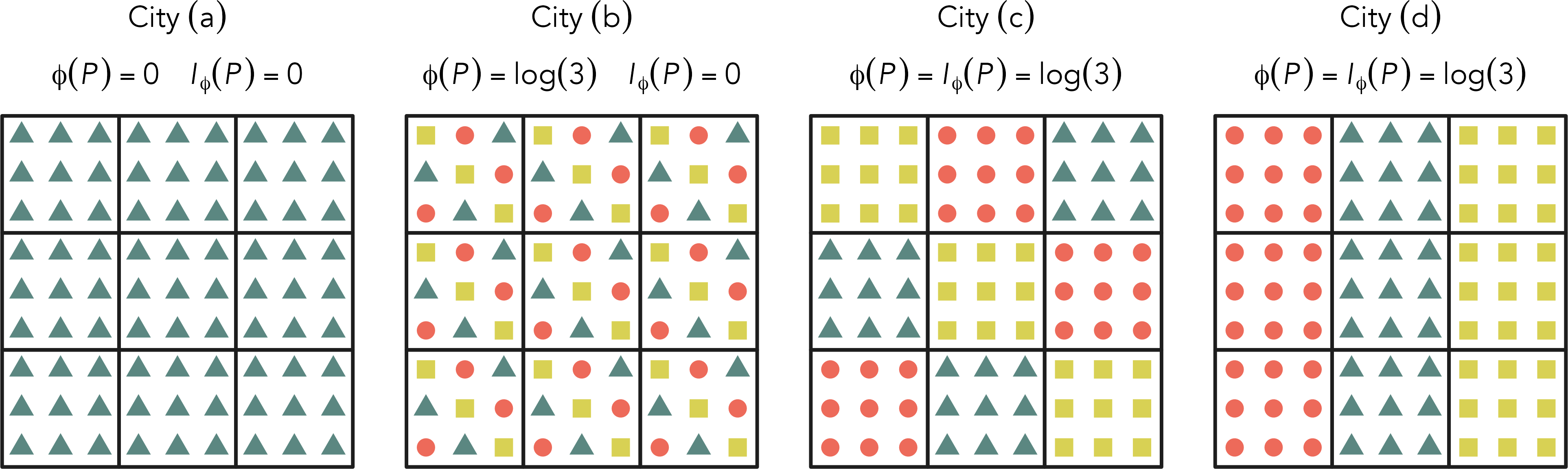}
  \caption{
    Synthetic ``cities'' with nine areal units, each of which contains a population of nine individuals in one of three demographic categories. 
    For each city, we give the Shannon entropy $\phi$ (measuring diversity) and Jensen information $I_\phi$ (measuring segregation); we define these quantities in \Cref{sec:diversity-convexity,sec:segregation}.
    } \label{fig:checkerboard}
\end{figure*}

Questions such as these have been studied for a sufficiently long time that a review article in 1955 could already express frustration with the chaotic state of the methodology. 
Commenting then on the methods available at the time, Duncan and Duncan \cite{duncanMethodologicalAnalysisSegregation1955} wrote: 
\begin{enumerate}
  \item \emph{The proposed indices of segregation have a number of hitherto unnoticed interrelationships which can be mathematically demonstrated.} 
  \item \emph{Some of them have mathematical properties of which their proponents were unaware, and which lead to difficulties of interpretation.} 
  \item \emph{As a consequence, the status of the empirical work already done with segregation indices is questionable, and their validity for further research is undermined.} 
\end{enumerate}
Since that time, there have been efforts within quantitative sociology to impose a more unified mathematical structure on segregation measures
\cite{reardonMeasuresMultigroupSegregation2002, reardonMeasuresOrdinalSegregation2009,reardon3MeasuresSpatial2004}.
There has also been some recent attention on segregation measurement in the mathematics community \cite{duchin2023measuring,topaz2025diversity,hoogstraQuantitativeMetricsTrait2025}. 
Despite the recent scholarly attention paid to segregation measurement, hints of the intellectual chaos described by Duncan and Duncan remain.
On the theme of ``\emph{hitherto unnoticed interrelationships},'' I am aware of two instances within the last ten years of a diversity or segregation measure being presented as novel despite already existing in the literature under a different name.

In this article, I present a mathematical perspective on segregation measurement conversant with formal work in the quantitative sociology literature. 
I argue that sticking close to this literature not only enriches opportunities for conversation between mathematicians and social scientists but also provides a fertile ground for interesting, accessible mathematics and data analysis.
This article is not a review and does not offer a comprehensive survey of segregation measurement; rather, it is an opinionated tour of one path through this constellation of topics. 
\Cref{sec:preliminaries} further describes the segregation measurement problem and establishes notation. 
In \Cref{sec:diversity-convexity}, I discuss diversity and argue that common intuitions lead us naturally to the class of convex functions defined on the probability simplex.
I then formulate segregation in \Cref{sec:segregation} as a \emph{Jensen information}: a comparison between global and average local diversity. 
I formalize several properties from the quantitative sociology literature and show that the Jensen information satisfies these properties.
In \Cref{sec:spatial-analysis}, I describe and illustrate three methods of conducting spatial analysis using the Jensen information: spatial smoothing, spatial information decomposition, and local Jensen information. 
\Cref{sec:discussion} closes with some discussion of future directions.   

\section{Preliminaries} \label{sec:preliminaries}

To motivate our mathematical formalism, we consider a set of simplified ``cities'' depicted in \Cref{fig:checkerboard}. 
In each of the four panels, individuals with one of three demographic traits are arranged into a synthetic ``city'' with nine areal units (``blocks''), each of which contains a population of nine individuals.
In City (a), only members of a single group are present; such a city can be neither diverse nor segregated. 
City (b) is diverse but not segregated; each group is represented uniformly across blocks.
Cities (c) and (d) do display spatial segregation.
In both, each block is demographically homogeneous. 
Furthermore, each city has exactly the same number of blocks dominated by the same groups.
Any measure of segregation which depends only on cross-tabulations of blocks and demographics will assign the same value of segregation to both panels. 
Such a measure may be called \emph{aspatial}. 
The ``checkerboard'' pattern of City (c), however, is arguably less segregated than (d) due to the spatial structure of the blocks: in (c), each block is surrounded by unlike blocks, while in (d) like blocks cluster together, producing monolithic homogeneous regions.  
The problem of constructing \emph{spatial} measures, which successfully distinguish between (c) and (d) in intuitive ways, has often been called the ``checkerboard problem'' in the quantitative sociology literature \cite{reardon3MeasuresSpatial2004}; as discussed above, such a measure must address the modifiable areal unit problem (MAUP) \cite{openshaw1984modifiable}.
Summarizing, a \emph{diversity} measure can distinguish panels (a) and (b); an aspatial \emph{segregation} measure can distinguish both from panel (c), and a \emph{spatial} segregation measure can distinguish panel (d) from panel (c).

We now establish notation. 
For $\ell \in \N$, let $[\ell] = \{1, 2, \ldots, \ell\}$. 
\begin{dfn} \label{def:simplex}
    For $\ell \in \N,$ the \emph{probability simplex} $\simplex{\ell}$ is the set 
    \begin{align*}
        \simplex{\ell} \eqdef \braces{\vy \in \R^{\ell} : \sum_{j = 1}^{\ell} y_j = 1, y_j \geq 0 \text{ for all } j \in [\ell]}
    \end{align*}
    of probability vectors on $\ell$ elements.\footnote{
      Although it is sometimes common to refer to $\simplex{\ell}$ as an $(\ell-1)$-simplex (this being its manifold dimension), we write $\simplex{\ell}$ to emphasize its role as the set of probability distributions over $\ell$ categories.
      } 
\end{dfn}
Let $\ve_{I} \in \simplex{\ell}$ be the vector with $1$ in the indices specified by the multi-index $I \subseteq [\ell]$ and $0$ elsewhere. 
Let $\relint{\simplex{\ell}}$ be the relative interior of $\simplex{\ell}$; if $\vy \in \relint{\simplex{\ell}}$, every entry of $\vy$ is strictly positive.

\begin{dfn} \label{def:population}
    A \emph{population} $\pop$ with $n$ locations in $d$ spatial dimensions and $\ell$ demographic categories is a triple $\pop = \standardTriple$, where: 
    \begin{itemize}
      \item $\mX \in \R^{n\times d}$ is a matrix with rows $\vx_i$, $i = 1, \ldots, n$.\footnote{In all our examples $n = 2$, though nothing in our results below formally depends on this. The matrix $\mX$ can be replaced by a discrete metric space $(\cX, d)$ where $d$ describes generalized distance between locations; under this replacement all the remaining results in this work hold except for \Cref{dfn:local-info,thm:local-info}, which rely on smooth approximations to discrete functions in Euclidean space.} 
      \item $\mY \in \R^{n \times \ell}$ is a row-stochastic matrix: for each row $i$, $\vy_i \in \simplex{\ell}$.
      \item $\vmu \in \simplex{n}$.   
    \end{itemize}
    We denote the set of all populations with $n$ locations, $d$ dimensions, and $\ell$ demographic categories as $\Pops_{nd\ell}$. 
\end{dfn}
In this definition, the matrix $\mX$ represents space: element $\vx_i$ of $\mX$ can be viewed as giving the location of the $i$th spatial subunit of the population. 
The matrix $\mY$ represents the demographic distribution of the population according to some attribute such as gender, race, or education level: the $i$th row $\vy_i$ of $\mY$ gives the proportion of individuals in each demographic category at location $i$.
Finally, the probability vector $\mu$ represents the population density across locations; the proportion of individuals living in location $i$ is $\mu_i$. 
The mean population vector $\bar{\vy} = \sum_{i = 1}^n \mu_i \vy_i$ has entries giving the marginal probabilities that an individual chosen uniformly at random from the entire population belongs to each demographic category. 
\begin{dfn} \label{def:population-measurement}
  A \emph{population measurement} is a function $f: \Pops_{nd\ell} \rightarrow \R$.
\end{dfn}
We aim to construct intuitive, tractable population measurements that capture natural intuitions about diversity and segregation.

\section{Diversity and Convexity} \label{sec:diversity-convexity}

A population is \emph{diverse} with respect to an attribute when it contains many categories of that attribute in roughly equal or equitable proportions.
Diversity is usually considered to be a property of a population as a whole, without regard for spatial subdivisions. 
Thus, a diversity measurement $f:\Pops_{nd\ell} \rightarrow \R$ is a function of the mean population demographic vector: $f(\pop) = \psi(\bar{\vy})$ for some function $\psi: \simplex{\ell} \rightarrow \R$.

Working in mathematical ecology, Pielou \cite{pielou1977mathematical} gives a proposed set of criteria for diversity measurements: 
\begin{enumerate}
  \item \textbf{Uniform maximization.} ``\emph{Diversity should be maximized when all groups are present in equal proportions.}'' 
  Formally, $\frac{1}{\ell} \ve_{[\ell]} = \argmax_{\vy \in \simplex{\ell}} \psi(\vy)$. 
  \item \textbf{Monotonicity under group addition.} ``\emph{In two populations in which all groups are equally represented, the population with the greatest number of groups will be more diverse.}'' 
  Formally, if $I\subset [\ell]$, then $\psi\paren{\ve_{I}/|I|} < \psi\paren{\ve_{[\ell]}/\ell}$.\footnote{Pielou also offers a third axiom, which is more restrictive and essentially characterizes a single diversity measure \cite{chakrabarti2005shannon}; we discuss it in the supplementary material.} 
\end{enumerate}
Pielou's proposed axioms capture important intuitions, but they don't offer a recipe for constructing useful diversity functions $\psi$. 
A hint is offered by monotonicity; we can think of a population uniformly distributed across $\ell$ demographic categories as a weighted mixture of two populations, one with $\abs{I}$ categories (where $I \subset [\ell]$) and another with $\ell - \abs{I}$ categories.
One candidate generalization of the group addition property is therefore that \emph{diversity should increase on average when populations are mixed}. 
Formally, let $\vy_1, \vy_2 \in \simplex{\ell}$ be two population demographic vectors and let $\lambda \in (0,1)$ be a mixture weight (expressing, for example, the relative population sizes). 
Mixing the populations forms the new demographic vector $\vy = \lambda \vy_1 + (1-\lambda) \vy_2$.
Our mixing intuition suggests that if $\vy_1 \neq \vy_2$, then 
\begin{align}
  \psi(\vy) > \lambda \psi(\vy_1) + (1-\lambda) \psi(\vy_2)\;. \label{eq:concave}
\end{align}
Here, the left-hand side gives the diversity of the mixed population $\vy$, while the right-hand side gives the $\lambda$-weighted average diversity of the original populations $\vy_1$ and $\vy_2$.
\Cref{eq:concave} leads us directly to the topic of convex functions. 
\begin{dfn}\label{def:convexity}
  Set $S\in \R^d$ is a \emph{convex set} if for all $\vy_1, \vy_2 \in S$ and $\lambda \in [0, 1]$, we have $\lambda \vy_1 + (1-\lambda) \vy_2 \in S$.
  If $S$ is a convex set, then $\phi: S \rightarrow \R$ is a \emph{convex function} if for all $\vy_1, \vy_2 \in S$ and $\lambda \in [0, 1]$, we have $\phi(\lambda \vy_1 + (1-\lambda) \vy_2) \leq \lambda \phi(\vy_1) + (1-\lambda) \phi(\vy_2)$.
  Function $\phi$ is \emph{strictly convex} if the above inequality is strict whenever $\vy_1 \neq \vy_2$ and $\lambda \in (0, 1)$. 
  Function $\psi$ is (\emph{strictly}) \emph{concave} if $\phi = -\psi$ is (strictly) convex.\footnote{
    Throughout the remainder of this article, we state without proof the convexity of certain functions as well as general properties of convex functions. 
    I recommend \cite{boydConvexOptimization2023} for a thorough treatment and proofs.
    A convenient approach to proving convexity for each of the following examples is to compute the Hessian $\nabla^2 \phi(\vy)$ and show that it is positive semidefinite for all $\vy \in \simplex{\ell}$ \cite{boydConvexOptimization2023}.}
\end{dfn}
The simplex $\simplex{\ell}$ is a convex set, and so the notion of convexity applies to functions defined on $\simplex{\ell}$.
We will abuse notation slightly and refer to a population measurement $f(\pop) = -\phi(\bar{\vy})$ as (\emph{strictly}) \emph{concave} if $\phi$ is (strictly) convex.\footnote{Following convention, we will emphasize convexity of $\phi$ rather than concavity of $-\phi$.}
Thus, our intuition expressed by \cref{eq:concave} is that a population diversity measurement should be strictly concave. 

\begin{figure}
  \centering
  \includegraphics[width=0.5\textwidth]{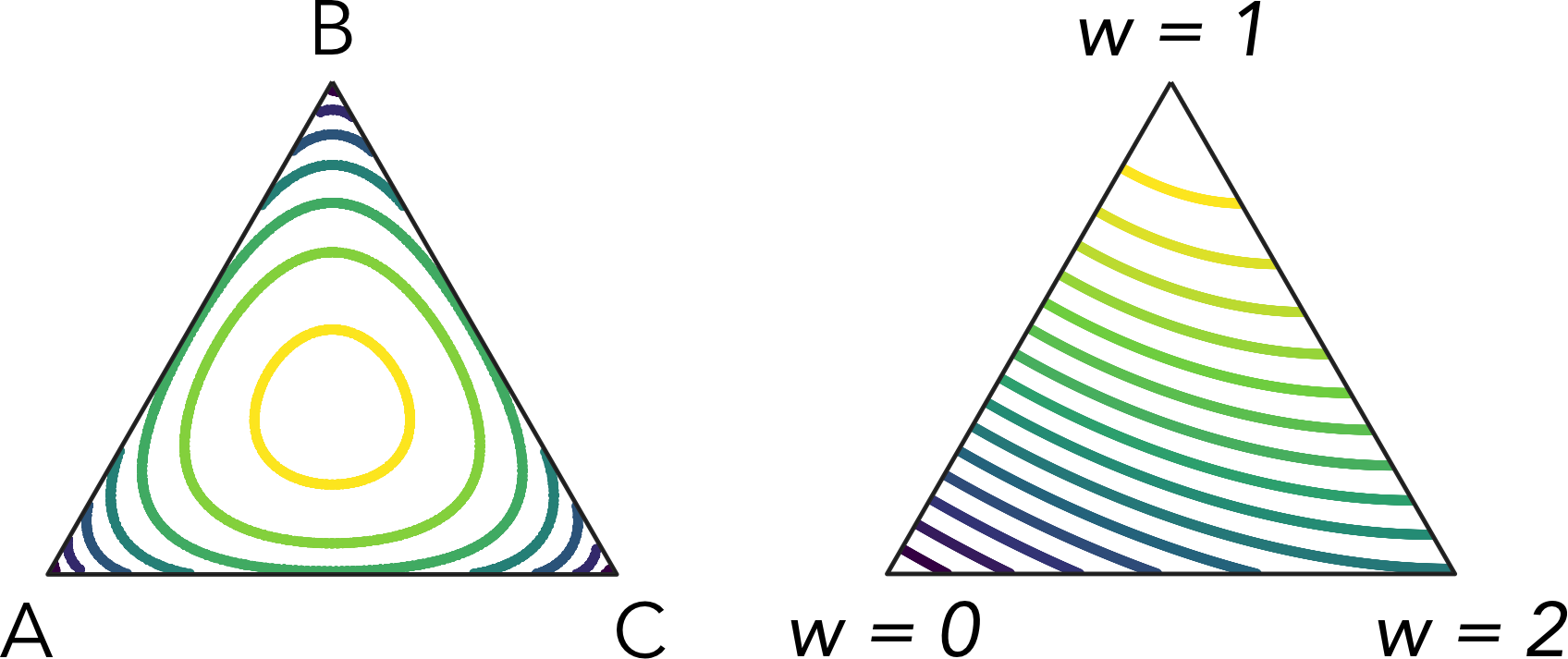}
  \caption{
    (Left): Isocontours of the Shannon entropy on the simplex $\simplex{3}$.
    (Right): Isocontours of the wealth distribution measurement from \Cref{example:equality} with possible wealth values $w \in \{1, 2, 3\}$ and reference distribution $\ve_1$ (corresponding to all individuals possessing 1 unit of wealth). 
    } \label{fig:isocontours}
\end{figure}

\begin{example}[Shannon Entropy] \label{example:shannon}
  If $\phi_H(\vy) = \sum_{j = 1}^\ell y_j \log y_j$, the resulting diversity measurement $\psi_H = -\phi_H$ is the \emph{Shannon entropy}, which is widely used in ecology, sociology, and information theory.\footnote{Throughout this article, we will use the notation $\log$ to denote the natural logarithm, as well as the standard convention that $0 \log 0 = 0$.} 
\end{example}

\begin{example}[Simpson Index] \label{example:simpson}
  If $\phi_S(\vy) = \sum_{j = 1}^\ell y_j^2 - 1$, the resulting diversity measurement $\psi_S = -\phi_S$ is the \emph{Simpson index}, which is widely used in ecology and sociology.
  The Simpson index has an attractive physical interpretation: it is the probability that two individuals chosen uniformly at random from the population belong to different demographic categories.
\end{example}

\begin{example}[A Bespoke Measurement of Wealth Distribution] \label{example:equality}
  For each $w \in 0,\ldots,w_\mathrm{max}$, let $y_{iw}$ give the proportion of individuals in location $i$ who possess $w$ units of wealth.
  Let $\mC$ be the  lower triangular matrix of ones: $\mC_{ij} = 1$ if $i \geq j$ and $\mC_{ij} = 0$ otherwise. 
  Then $\mC \vy_i$ gives the cumulative distribution vector of wealth at location $i$. 
  Fix a reference distribution vector $\vyref$. 
  The function $\phi_{\mathrm{eq}}(\vy) = \norm{\mC (\vy - \vyref)}_2^2$ is strictly convex, and so $f(\pop) = -\norm{\mC (\bar{\vy} - \vyref)}_2^2$ is a concave population measurement.
  This population measurement expresses a normative preference for the distribution of wealth corresponding to $\vyref$; it is maximized when $\bar{\vy} = \vyref$. 
  A value such as wealth equality can be encoded by setting $\vyref$ to be the $\ve_{\bar{w}}$ in which all individuals possess the same amount of wealth $\bar{w}$.
\end{example}

We can now address Pielou's axioms for diversity measures: 

\begin{thm} \label{thm:monotone-under-group-addition}
  If $\phi: \simplex{\ell} \rightarrow \R$ is strictly convex, then the population measurement $f(\pop) = -\phi(\bar{\vy})$ satisfies Pielou's monotonicity property. 
\end{thm}

All three of the diversity functions in \Cref{example:shannon,example:simpson,example:equality} are strictly convex and therefore satisfy monotonicity.
The uniform maximizer property, however, requires additional conditions. 
  \begin{dfn}
    Function $\phi$ is \emph{symmetric} if for all $\vy \in \simplex{\ell}$ and all permutations $\sigma$ of $[\ell]$, we have $\phi(\vy) = \phi(\sigma(\vy))$, where $\sigma(\vy)$ is the vector obtained by permuting the entries of $\vy$ according to $\sigma$.
    Function $\phi: \simplex{\ell} \rightarrow \R$ is \emph{separable} if there exists a function $g: [0, 1] \rightarrow \R$ such that $\phi(\vy) = \sum_{j = 1}^{\ell} g(y_j)$ for all $\vy \in \simplex{\ell}$.\footnote{Separability implies symmetry, but not conversely.}
  \end{dfn}

  \begin{thm}\label{thm:diversity-axioms-1}
    If $\phi: \simplex{\ell} \rightarrow \R$ is strictly convex and symmetric, then the population measurement $f(\pop) = -\phi(\bar{\vy})$ satisfies Pielou's uniform maximization property.
  \end{thm}
  Both the Shannon entropy given by $\phi_{H}$ and the Simpson index given by $\phi_{S}$ are separable, therefore symmetric, and thus also satisfy the uniform maximizer property, while the wealth measurement $\phi_{\mathrm{eq}}$ is not symmetric and does not.
  We can see these results visually in \Cref{fig:isocontours}.

\section{Segregation as Information} \label{sec:segregation}

Examining the checkerboard ``cities'' (b-d) of \Cref{fig:checkerboard}, we see that in city (b) each location is exactly as diverse as the population as a whole ($\vy_i = \vy$ and thus $\phi(\vy_i) = \phi(\vy)$ for each $i$), while in cities (c) and (d) each location is less diverse than the population as a whole.
This suggests a simple quantitative segregation measurement: the difference between global and mean local diversities. 
\begin{dfn}[Jensen Information \cite{banerjeeClusteringBregmanDivergences2005}]
  Given a concave diversity function $f = -\phi$,  the \emph{Jensen information}  of a population $\pop = \standardTriple \in \Pops_{nd\ell}$ is
  \begin{align}
    \info_{\phi}(\pop) &\eqdef f(\bar{\vy}) - \sum_{i = 1}^n \mu_i f(\vy_i) \nonumber \\ 
    &= \sum_{i = 1}^n \mu_i \brackets{\phi(\vy_i) - \phi(\bar{\vy})} \;.  \label{eq:jensen-information}
  \end{align}
\end{dfn}
Jensen's Inequality, a standard result in convex analysis, implies that $\info_{\phi}(\pop) \geq 0$ for all populations $\pop$ and convex functions $\phi$. 
If $\phi$ is strictly convex, then $\info_{\phi}(\pop) = 0$ if and only if $\vy_i = \bar{\vy}$ for all $i \in [n]$.

\begin{example}[Shannon Mutual Information]
  The Jensen information $\info_{\phi_H}(\pop)$ corresponding to the negative Shannon entropy $\phi_H$ is the \emph{mutual information} between the random variables $X$ and $Y$ describing location and demographic category of a randomly chosen individual from the population: 
  \begin{align}
    \info_{\phi_H}(\pop) = \sum_{i = 1}^n \sum_{j = 1}^{\ell} \mu_i y_{ij} \log \frac{y_{ij}}{\bar{y}_j}\;. \label{eq:shannon-jensen-info}
  \end{align}
  The mutual information has an attractive interpretation as the extent to which knowledge of an individual's location reduces information-theoretic uncertainty about their demographic category, and vice versa \cite{cover2012elements}.
\end{example}

\begin{example}[Simpson Information]
  The Jensen information $\info_{\phi_S}(\pop)$ corresponding to the Simpson diversity index $\phi_S$ is the \emph{Simpson information} of the population: 
  \begin{align}
    \info_{\phi_S}(\pop) = \sum_{i = 1}^n \mu_i \norm{\vy_i}_2^2 - \norm{\bar{\vy}}_2^2 \;. \label{eq:simpson-jensen-info}
  \end{align}
\end{example}

\begin{example}[Wealth Information]
  Let $\phi_{\mathrm{eq}}(\vy) = \norm{\mC (\vy - \vyref)}_2^2$ be the wealth equality function from \Cref{example:equality}. 
  Then, the Jensen information $\info_{\phi_{\mathrm{eq}}}(\pop)$ is given by 
  \begin{align}
    I_{\phi_{\mathrm{eq}}}(\pop) = \sum_{i = 1}^n \mu_i \norm{\mC\vy_i}_2^2 - \norm{\mC\bar{\vy}}_2^2 \;. \label{eq:wealth-jensen-info}
  \end{align}
  Importantly, this Jensen information is independent of the reference distribution $\vyref$. A derivation is provided in the supplementary material. 
\end{example}

Although the Jensen information in \cref{eq:jensen-information} is nonnegative, individual terms of the form $\phi(\vy_i) - \phi(\bar{\vy}) $ can be negative. 
This means that a given location's ``contribution to segregation'' cannot simply be interpretably measured from the terms in \cref{eq:jensen-information}.
It is sometimes useful for analytic purposes to write the Jensen information as a weighted sum of nonnegative contributions from each location: 
\begin{align}
  \info_{\phi}(\pop) = \sum_{i = 1}^n \mu_i g(\vy_i, \bar{\vy}) \;. \label{eq:information-equivalence}
\end{align}
Under mild regularity conditions, this is possible in a unique way: 
\begin{thm}
  Let $\phi: \simplex{\ell} \rightarrow \R$ be a strictly convex function which is differentiable on $\relint{\simplex{\ell}}$. 
  Then, \cref{eq:information-equivalence} holds for some function $g:\relint{\simplex{\ell}} \times \simplex{\ell} \rightarrow \R$ and for any population $\pop \in \Pops_{nd\ell}$ if and only if $g$ is the \emph{Bregman divergence} $\divergence{\phi}{\cdot}{\cdot}$ induced by $\phi$: 
  \begin{align}
    \divergence{\phi}{\vy}{\bar{\vy}} \eqdef \phi(\vy) - \phi(\bar{\vy}) - \nabla \phi(\bar{\vy})^{\top} (\vy - \bar{\vy}) \;. \label{eq:bregman-divergence}
  \end{align}
\end{thm}
The ``if'' direction of this theorem is a standard result by direct calculation \cite{banerjeeClusteringBregmanDivergences2005}, while the ``only if'' direction was recently shown by the present author \cite{chodrowEquivalenceInformationsCharacterizes2025b}. 
The Bregman divergence is a measure of the difference between the value of $\phi$ at $\vy$ and the value of the linearization (first-order Taylor expansion) of $\phi$ around $\bar{\vy}$ evaluated at $\vy$.
The Bregman divergence can be interpreted as a measure of dissimilarity of $\vy$ from $\bar{\vy}$.  
It is always nonnegative and is zero if and only if $\vy = \bar{\vy}$.
The Bregman divergence is not a metric, as it is not symmetric and does not satisfy the triangle inequality, although it does satisfy many other useful properties as a comparison function \cite{banerjeeClusteringBregmanDivergences2005}. 
In the case of the Shannon mutual information \eqref{eq:shannon-jensen-info}, the Bregman divergence is the Kullback-Leibler divergence $g(\vy, \bar{\vy}) = \sum_{j = 1}^{\ell} y_j \log(y_j / \bar{y}_j)$. 
In the case of the Simpson information \eqref{eq:simpson-jensen-info}, the Bregman divergence is the squared Euclidean distance $g(\vy, \bar{\vy}) = \norm{\vy - \bar{\vy}}_2^2$.
Finally, in the case of the wealth information \eqref{eq:wealth-jensen-info}, the Bregman divergence is $g(\vy, \bar{\vy}) = \norm{\mC (\vy - \bar{\vy})}_2^2$, the Mahalanobis distance with respect to the matrix $\mC^T\mC$.
\Cref{fig:simplex-information} gives a geometric view of \cref{eq:information-equivalence}.  

 \begin{figure}
    \centering
    \includegraphics[width=0.5\textwidth]{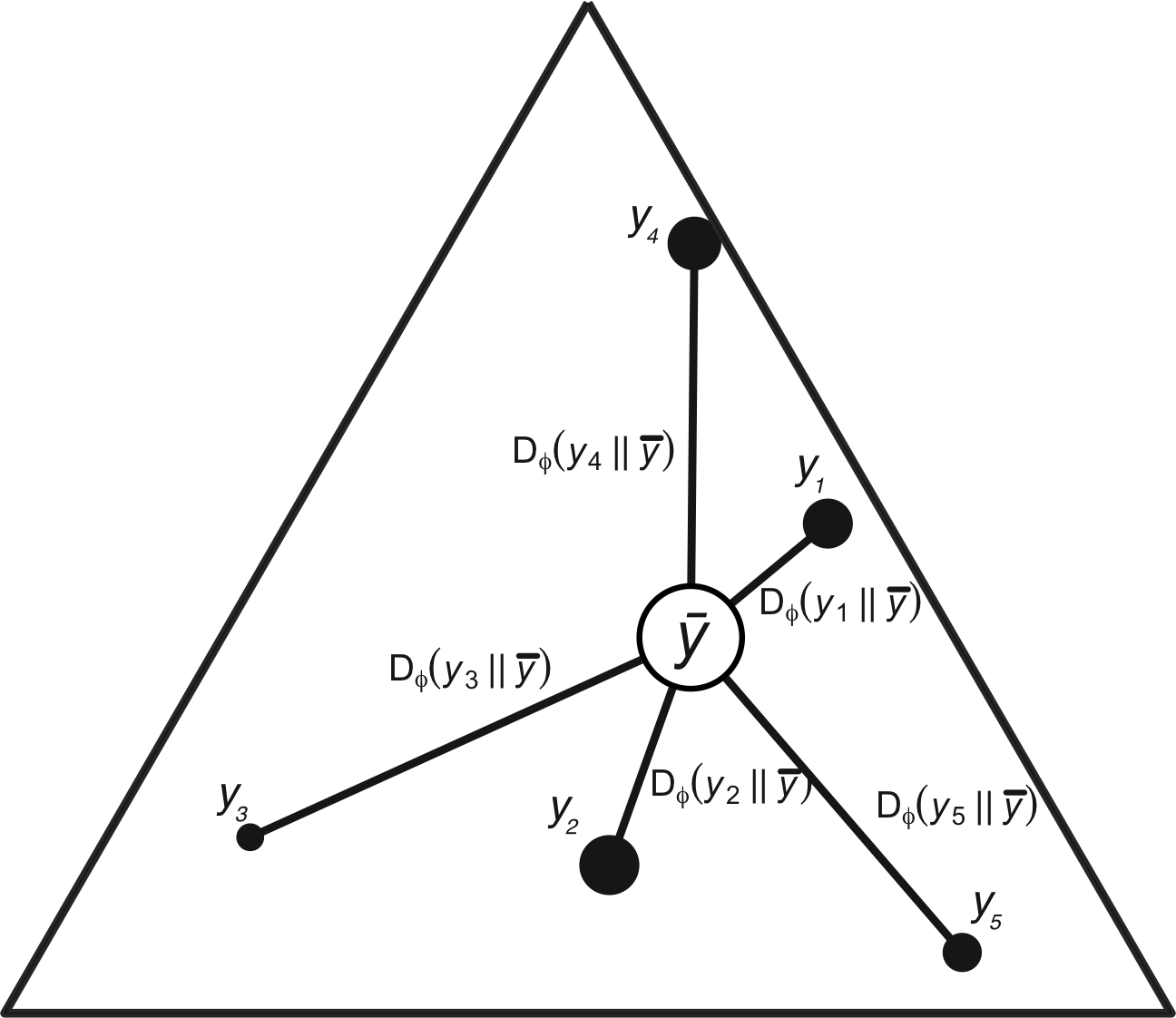}
    \caption{
      Geometric illustration of \cref{eq:information-equivalence}. 
      Points $\vy_i$, $i = 1,\ldots, 5$ are sized according to their weight $\mu_i$ and are located in the simplex $\simplex{3}$ according to their demographic distribution.
      The Jensen information $\info_\phi(\pop)$ is the weighted mean of the Bregman divergences $\divergence{\phi}{\vy_i}{\bar{\vy}}$ of each point from the mean $\bar{\vy}$.
      The visual intuition relies on viewing the Bregman divergence as ``distance-like,'' though it should be remembered that it is not a metric. 
    }\label{fig:simplex-information}
  \end{figure}

A useful property of Bregman divergences is that they are \emph{locally} quadratic: for any $\vy \in \relint{\simplex{\ell}}$ and $\vdelta$ in the tangent space to $\simplex{\ell}$ at $\vy$, we have
\begin{align}
  \divergence{\phi}{\vy + \vdelta}{\vy} = \frac{1}{2} \vdelta^{\top} \hess{\phi}(\vy) \vdelta + o(\norm{\vdelta}_2^2)  \;, \label{eq:local-quadratic}
\end{align}
where $\hess{\phi}(\vy)$ is the Hessian of $\phi$ at $\vy$ and $o(\norm{\vdelta}_2^2)$ is a term which is asymptotically smaller than $\norm{\vdelta}_2^2$ as $\norm{\vdelta}_2 \rightarrow 0$.\footnote{ 
A heuristic derivation proceeds by expanding $\phi(\vy + \vdelta)$ about $\vy$ and noting that the other terms in \cref{eq:bregman-divergence} cancel the zeroth and first-order terms in the expansion.
} 
Since $\phi$ is strictly convex, the Hessian $\hess{\phi}(\vy)$ is positive definite for all $\vy \in \relint{\simplex{\ell}}$, and so the leading term in the expansion is a positive definite quadratic form.

\subsection{The Jensen Information as a Segregation Measure} \label{sec:segregation-properties}

In an influential article, Reardon and Firebaugh \cite{reardonMeasuresMultigroupSegregation2002} propose a set of desirable criteria for multigroup segregation measures, and assess a number of existing measures against these criteria with proofs and counterexamples.
Of the six measures evaluated by the authors, exactly two of these are Jensen informations, corresponding to scaled versions of the Shannon mutual information\footnote{In the segregation measurement setting the mutual information is sometimes referred to as the \emph{Theil index} or the \emph{information theory index}.} induced by the negative Shannon entropy $\phi_H$ and the Simpson index induced by the square norm $\phi_S$.
These two measurements come out on top in satisfying the various criteria proposed by the authors.
In particular, Reardon and Firebaugh ``...\emph{conclude that the [Shannon mutual information] is the most conceptually and mathematically satisfactory index}...''. 
It turns out this isn't a coincidence: many desirable intuitions about segregation follow from convexity.\footnote{Indeed, a set of relatively natural desiderata for a segregation measurement, including symmetry and a generalized monotonicity property, is sufficient to uniquely determine the Shannon information \cite{frankelMeasuringSchoolSegregation2011}.} 

Some of the Reardon and Firebaugh criteria are trivially satisfied by any Jensen information; for example, their \emph{size invariance} criterion expresses simply the desire that $f$ be a function only of the joint population \emph{distribution} and not absolute population counts. 
Here we'll state one of the nontrivial properties and give sufficient criteria for a Jensen information to satisfy it.
The intuition expressed by the \emph{exchanges} property is that when individuals move from locations where they are relatively overrepresented to locations where they are relatively underrepresented, the population becomes less segregated.
Let $\tilde{\ve}_j \in \R^{\ell}$ denote a standard basis vector in $\R^{\ell}$. 

\begin{dfn}
  An \emph{exchange} of magnitude $\epsilon$ of groups $a$ and $b$ between locations $i$ and $j$ is a map $\exc{a}{b}{i}{j}{\epsilon}: \Pops_{nd\ell} \rightarrow \Pops_{nd\ell}$ which takes a population $\pop = \standardTriple$ and returns a new population $\pop' = (\mX, \mY', \mu')$ such that
  \begin{align*}
    \vmu'  &= \vmu, \\ 
    \vy_i' &= \vy_i + \frac{\epsilon}{\mu_i} (\tilde{\ve}_b - \tilde{\ve}_a), \\
    \vy_j' &= \vy_j + \frac{\epsilon}{\mu_j} (\tilde{\ve}_a - \tilde{\ve}_b). 
  \end{align*}
\end{dfn}

\begin{dfn}
  A segregation measurement $S$ satisfies the \emph{exchanges property of Reardon and Firebaugh} \cite{reardonMeasuresMultigroupSegregation2002} if whenever $y_{i a} < y_{j a}$ and $y_{i b} > y_{j b}$, we have $S(\exc{a}{b}{i}{j}{\epsilon}(\pop)) < S(\pop)$ for all $\epsilon$ sufficiently small. 
\end{dfn}

\begin{thm} \label{thm:exchanges}
  If $\phi$ is separable and strictly convex, then the Jensen information $\info_{\phi}$ satisfies the exchanges property of Reardon and Firebaugh.
\end{thm}

\Cref{thm:exchanges} implies that the Shannon mutual information and the Simpson information, both being generated by separable functions $\phi$, satisfy the exchanges property. 
Since the wealth equality function $\phi_{\mathrm{eq}}$ is not separable, \Cref{thm:exchanges} gives no such guarantee. 


\section{Spatial Analysis with the Jensen Information} \label{sec:spatial-analysis}

  Somewhat figuratively, the problem of ``spatializing'' a segregation measure is the problem of distinguishing the checkered segregation pattern of \Cref{fig:checkerboard}(c) and the monolithic segregation pattern of \Cref{fig:checkerboard}(d).
  There are several intuitions that can help us distinguish the two patterns: 
  \begin{itemize}
    \item \textbf{Perturbation}: If we imagine adding random noise by randomly moving individuals to nearby locations, checkered city (c) will tend to become integrated more rapidly than monolithic city (d). 
    \item \textbf{Scale}: The checkered city (c) has a smaller \emph{spatial scale} of segregation, in the sense that more contiguous regions are required to partition the population into homogeneous subpopulations.
    \item \textbf{Boundaries}: The checkered city (c) displays more \emph{boundaries} between regions of different demographic composition. 
  \end{itemize}
  These intuitions are interrelated, but can motivate different directions for segregation measurement. 
  All of them may be viewed as attempts to construct spatially-aware segregation measurements which are robust to the Modifiable Areal Unit Problem (MAUP) \cite{openshaw1984modifiable}.
  In this section, we survey one approach in each direction and illustrate with computational experiments. 
  In each, we apply the Shannon mutual information $I_{\phi_H}$ as our segregation measurement to demographics of American cities using race and ethnicity data from the 2024 American Community Survey (ACS), which we aggregate into the same five categories used in \Cref{fig:dot-viz}.

  \subsection{A Perturbative View: Spatial Smoothing} \label{sec:smoothing}

    \begin{figure}
      \centering 
      \includegraphics[width=0.5\textwidth]{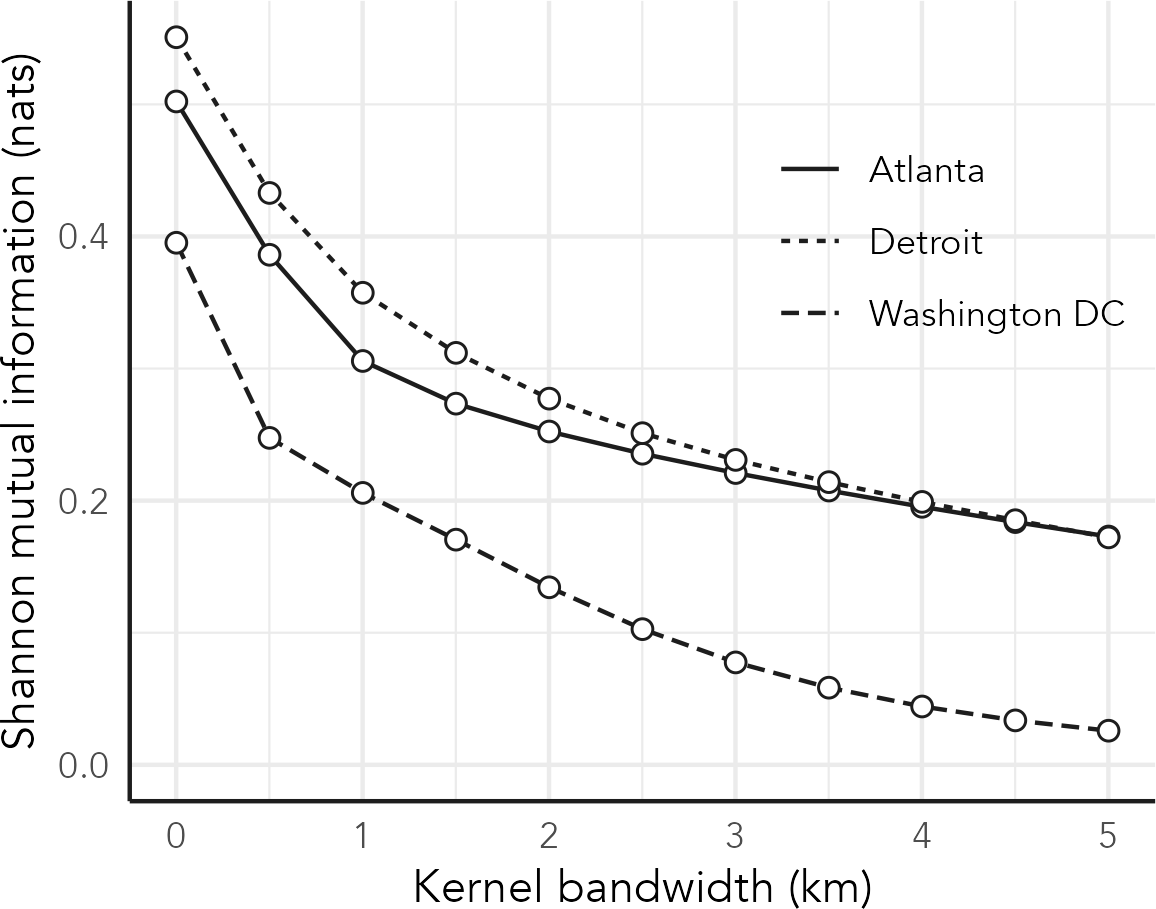}
      \caption{
        Illustration of a perturbative approach to spatial segregation measurement for Atlanta, Georgia; Detroit, Michigan; and Washington D.C.
        We show spatial profile curves for segregation under Gaussian Radial Basis (RBF) smoothing as the bandwidth $\gamma$ is varied. 
      }
      \label{fig:smoother}
    \end{figure}

    One approach to spatialization, now standard in the quantitative sociology literature \cite{reardon3MeasuresSpatial2004}, is to preprocess the population data by applying a spatial smoothing kernel. 
    This amounts to computing the composition $\info_\phi \circ \smooth_\gamma$ of the Jensen information with a smoothing operator $\smooth_\gamma: \pop \rightarrow \pop'$ parameterized by a bandwidth $\gamma$.
    In a typical setup, $\mY'$ is obtained by applying a smoothing kernel to the rows of $\mY$: for each location $i$, 
    \begin{align*}
      \vy_i' = \frac{1}{\sum_{j = 1}^n k_\gamma(\vx_i, \vx_j)} \sum_{j = 1}^n k_\gamma(\vx_i, \vx_j) \vy_j \;,
    \end{align*}
    where $k_\gamma: \R^d \times \R^d \rightarrow [0, \infty)$ is a pairwise smoothing kernel. 
    Possible choices of the smoothing kernel $k_\gamma$ include parametric kernels such as the Gaussian radial basis function $k_\gamma(\vx_i, \vx_j) = \exp \paren{-\norm{\vx_i - \vx_j}_2^2 / 2\gamma^2}$, or data driven kernels based on road network distance or commute time between locations \cite{robertoSpatialProximityConnectivity2018}. 
    These smoothing kernels may be viewed as continuum limits of a random perturbation in which we randomly move individuals from their original location to nearby locations with probability proportional to the kernel value $k_\gamma$.
    Increasing the bandwidth parameter $\gamma$ will tend to reduce the impact of the Modifiable Areal Unit Problem (MAUP) by reducing sensitivity to the original areal units, but will also tend to reduce the measured segregation.

    In their theoretical analysis of spatialized segregation measures, Reardon and O'Sullivan \cite{reardon3MeasuresSpatial2004} find that spatialized measures corresponding to Shannon and Simpson informations, again the only Jensen informations amongst their considered measures, satisfy the greatest number of their proposed criteria for spatialized segregation measures.  
    This follows largely by inheriting the favorable properties of the general class of Jensen informations. 

    \Cref{fig:smoother} illustrates the smoothing-based approach on three sample American cities: Atlanta, Georgia; Detroit, Michigan; and Washington D.C. 
    We apply Gaussian radial basis smoothing with varying bandwidth $\gamma$ prior to computing the information $\info_{\phi_H}$.  
    The results highlight the relative spatial scales of segregation in each city; whereas Atlanta and Detroit have relatively similar curves which display persistent segregation even at $\gamma = 5\;\mathrm{km}$, Washington D.C. displays a much more rapid drop in segregation as the smoothing bandwidth is increased.
    In this sense, we might say that Washington D.C. displays a ``finer'' scale of segregation than Atlanta or Detroit.
    It is important to interpret this result in context, however: this result arguably as much reflects the smaller relative size of Washington D.C. as a whole in comparison to Atlanta or Detroit.

  \subsection{A Scale-Based View: Spatial Information Decomposition} \label{sec:decomposition}

    \begin{figure*}
      \centering 
      \includegraphics[width=\textwidth]{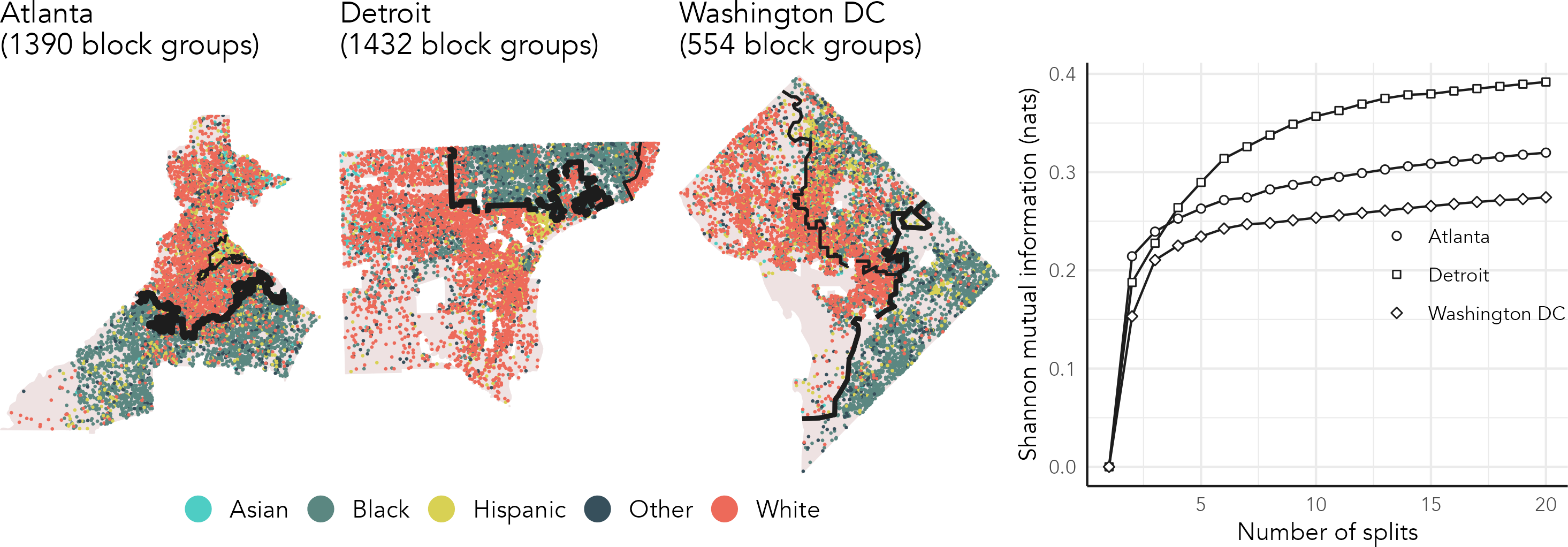}
      \caption{
        First and second learned major boundaries (heavy black lines) and information loss associated to them for each of three sample cities. 
        Each map contains 10,000 dots reflecting variable-sized units of population of the corresponding demographic category. 
      }\label{fig:aggregation}
    \end{figure*}

    Another response to the MAUP is to consider the question of how a segregation measure should change as we modify the areal units under study. 
    We cannot change the areal units arbitrarily, but we can \emph{aggregate} them into larger units. 
    The behavior of a segregation measure under aggregation can then be used to describe the spatial scale of segregation in a population.
    Returning to the grid-cities in \Cref{fig:checkerboard}, for example, it is possible to aggregate the columns of City (d) into three spatially contiguous, homogeneous, monolithic regions. 
    We can say that City (d) admits a spatially contiguous aggregation in which all segregation remains \emph{between} aggregated units and none is captured \emph{within} units. 
    No such aggregation is possible for City (c). 
    The Jensen information provides a helpful, formal decomposition into between- and within-aggregated-unit segregation.
    
  \begin{dfn}[Spatial Aggregation]
    Let $n' < n$ and let $\pi: [n] \rightarrow [n']$ be a surjective map.
    A \emph{spatial aggregation} of a population $\pop = \standardTriple \in \Pops_{nd\ell}$ by $\pi$ is a new population $\pop' = (\mX', \mY', \vmu') \in \Pops_{n'd'\ell}$ such that \footnote{The definition of $\mX'$ under aggregation is intentionally flexible; one possible choice is to set $x'_j = \frac{1}{\mu'_j} \sum_{i \in \pi^{-1}(j)} \mu_i x_i$, the $\mu$-weighted centroid of the locations in the $j$th aggregated unit.}
    \begin{align*}
      \mu'_j &= \sum_{i \in \pi^{-1}(j)} \mu_i, &\quad 
      \vy'_j &= \frac{1}{\mu'_j} \sum_{i \in \pi^{-1}(j)} \mu_i \vy_i\;. 
    \end{align*}
  \end{dfn}

  \begin{dfn}[Conditional Jensen Information]
    Let $\pop'$ be a spatial aggregation of $\pop$ by $\pi$.
    The \emph{conditional Jensen information} of $\pop$ given $\pop'$ is 
    \begin{align*}
      I_{\phi}(\pop \mid \pop') = \sum_{j = 1}^{n'} \mu_j' \underbrace{\paren{\phi(\vy_j') - \frac{1}{\mu_j'}\sum_{i \in \pi^{-1}(j)} \mu_i \phi(\vy_i)}}_{I_\phi(\pop_j)}   \;, 
    \end{align*}
    where $\pop_j$ is the subpopulation of $\pop$ corresponding to the $j$th aggregated unit.
  \end{dfn}
  The conditional Jensen information is the mean segregation \emph{within} each of the aggregated subpopulations, weighted by the population size of each.
  It follows that $I_{\phi}(\pop \mid \pop') \geq 0$ for all populations $\pop$ and spatial aggregations $\pop'$, with equality if and only if the demographic distribution is uniform within each aggregated unit.
  A direct calculation shows 
  \begin{align}
    I_\phi(\pop) = I_\phi(\pop') + I_\phi(\pop \mid \pop') \;. \label{eq:information-decomposition}
  \end{align}
  \Cref{eq:information-decomposition} is a general decomposition into ``between aggregated unit'' segregation (first term) and ``within aggregated unit'' segregation (second term).
  
  The behavior of the Jensen information under spatial aggregation offers us another route towards addressing the modifiable areal unit problem (MAUP). 
  To use it, however, we must address two challenges.
  First, we remain at the mercy of the highest-resolution data available. 
  Second and more fundamentally, the choice of aggregation $\pi$ is often \emph{itself} arbitrary, frequently based on nested administrative boundaries such as Census block groups, tracts, and counties.
  One approach to obtaining non-arbitrary aggregations is to seek spatially sensible aggregations which maximize the first, between-unit segregation term in \cref{eq:information-decomposition} while minimizing the second. 
  The task of forming such spatial aggregations is often called \emph{regionalization}; Alec Kirkley offers a helpful recent overview as well as a novel method \cite{kirkleySpatialRegionalizationBased2022}. 
  In my previous work, I proposed a simple agglomerative hierarchical clustering algorithm which greedily combines adjacent areal units which least decrease the between-unit segregation term in \cref{eq:information-decomposition} \cite{chodrowStructureInformationSpatial2017}.
  The resulting algorithm has the interpretive benefit that boundaries between aggregated areal units can be assigned information content associated with their contribution to the between-unit segregation term. 
  
  \Cref{fig:aggregation} illustrates this algorithm on our three sample cities, again using the Shannon mutual information $\info_{\phi_H}$ as the segregation measure.
  In each, the thickness of the border lines is proportional to their information content. 
  We also show a second kind of profile curve, which this time measures the amount of information contained in the aggregation as we increase the number of aggregated units.
  Coarse-structured segregation reminiscent of \Cref{fig:checkerboard}(d) is reflected by high information content expressed using a very small number of boundaries. 
  Care is required in interpreting these boundaries considering that the greedy hierarchical clustering algorithm used to create them has no optimality guarantees. 
  An important difference between this style of analysis and the smoothing profile curves of \Cref{fig:smoother} is that the aggregation-based analysis is relatively insensitive to the geographic size of the region analyzed; this may be a virtue or a drawback depending on analytic context. 
  
  \subsection{A Boundary-Based View: Local Information}
  
    A third approach to spatializing segregation measurements follows the intuition that the ``checkerboard'' pattern of \Cref{fig:checkerboard}(c) has more \emph{boundaries} or \emph{transitions} between demographically contrasting regions than the monolithic pattern of \Cref{fig:checkerboard}(d).
    To formalize this idea, it is convenient to model the observed demographic distribution via a differentiable function $\Upsilon: \cR \rightarrow \simplex{\ell}$, where an open set $\cR \subset \R^d$ represents continuous geographic space and $\vy = \Upsilon(\vx)$ gives the demographic distribution at location $\vx \in \cR$.
    This modeling assumption is arguably unrealistic given the presence of sharp boundaries in real-world settled environments, but can be enforced in data analysis through the application of a differentiable smoothing kernel to the demographic data as described in \Cref{sec:smoothing}.
    Under this assumption, we can describe the sensitivity of demographics $\Upsilon$ to space through the Jacobian $\jacobian{\Upsilon}(\vx) \in \R^{\ell \times d}$.
    One could in principle define the ``boundariness'' of $\Upsilon$ at a point $\vx \in \relint{\cR}$ via a measure of the magnitude of the Jacobian such as $b(\vx) = \norm{\jacobian{\Upsilon}(\vx)}_1 \eqdef \sum_{i=1}^{\ell} \sum_{j=1}^d \abs{\partial \Upsilon_i / \partial x_j}$. 
    An aggregate measure of ``boundariness''  over the entire region $\cR$ is then given by the total variation $B(\cR) = \int_{\cR} b(\vx) d\vx$.
    This straightforward approach has an important limitation: it is insensitive to the choice of diversity function $\phi$ and Jensen information $I_\phi$ used in analysis. 

    We can obtain a more flexible class of measurements by considering the behavior of the Jensen information $I_\phi$ within a small neighborhood about a point $\vx \in \relint{\cR}$. 
    We start by defining a continuous Jensen information: 
    \begin{align}
      I_\phi(\Upsilon, \cR) \eqdef \int_{\cR} \brackets{\phi\paren{\Upsilon(\vx)} - \phi\paren{\bar{\vy}(\cR)}} \; d\mu(\vx) \;,
    \end{align}
    where $\bar{\vy}(\cR) = \int_{\cR} \Upsilon(\vx) \; d\mu(\vx)$ is the mean demographic distribution over the region $\cR$.
    \begin{dfn} \label{dfn:local-info}
      The \emph{local information} of $\Upsilon$ at a point $\vx \in \relint{\cR}$ is 
      \begin{align}
        j_\phi(\Upsilon, \vx) \eqdef \lim_{\epsilon \rightarrow 0} \frac{I_\phi(\Upsilon, B_\epsilon(\vx))}{\epsilon^2}\;,  \label{eq:local-info-def}
      \end{align}
      where $B_\epsilon(\vx) \eqdef \{\vy \in \cR : \norm{\vy - \vx}_2 < \epsilon\}$ is the open ball of radius $\epsilon$ centered at $\vx$.
    \end{dfn}
    Intuitively, the local information should be near 0 in a small region of near-constant demographic composition, and should be large in a small region where the demographic composition changes rapidly.
    We can either calculate the local information at many points for use in visualization, or integrate the local information over the entire region to obtain an aggregate spatialized measure, the \emph{mean local information}: 
    \begin{align}
      J_\phi(\Upsilon, \cR) \eqdef \int_{\cR} j_\phi(\Upsilon, \vx) \;d\mu(\vx) \;.  \label{eq:mean-local-info}
    \end{align}
    Unlike the total variation measure $B(\cR)$ described above, the mean local information $J_\phi(\Upsilon, \cR)$ is sensitive to the choice of diversity function $\phi$ used in analysis.
    To use \cref{eq:mean-local-info}, however, we must know that the limit in \cref{eq:local-info-def} exists and must be able to approximate it in a computationally tractable way.
    \begin{thm}[\cite{chodrowStructureInformationSpatial2017}] \label{thm:local-info}
      Suppose that $\Upsilon: \cR \rightarrow \simplex{\ell}$ is differentiable on $\relint{\cR}$. 
      Fix $\vx \in \cR$ and let $\vy = \Upsilon(\vx)$. 
      Then, the limit in \cref{eq:local-info-def} exists and is given by 
      \begin{align}
        j_\phi(\Upsilon, \vx) = \frac{1}{8\pi} \trace\paren{\jacobian{\Upsilon}(\vx)^{\top} \hess{\phi}(\vy)  \jacobian{\Upsilon}(\vx)} \;.  \label{eq:local-info-trace}
      \end{align}
    \end{thm}
    
    \begin{figure}
      \includegraphics[width=0.45\textwidth]{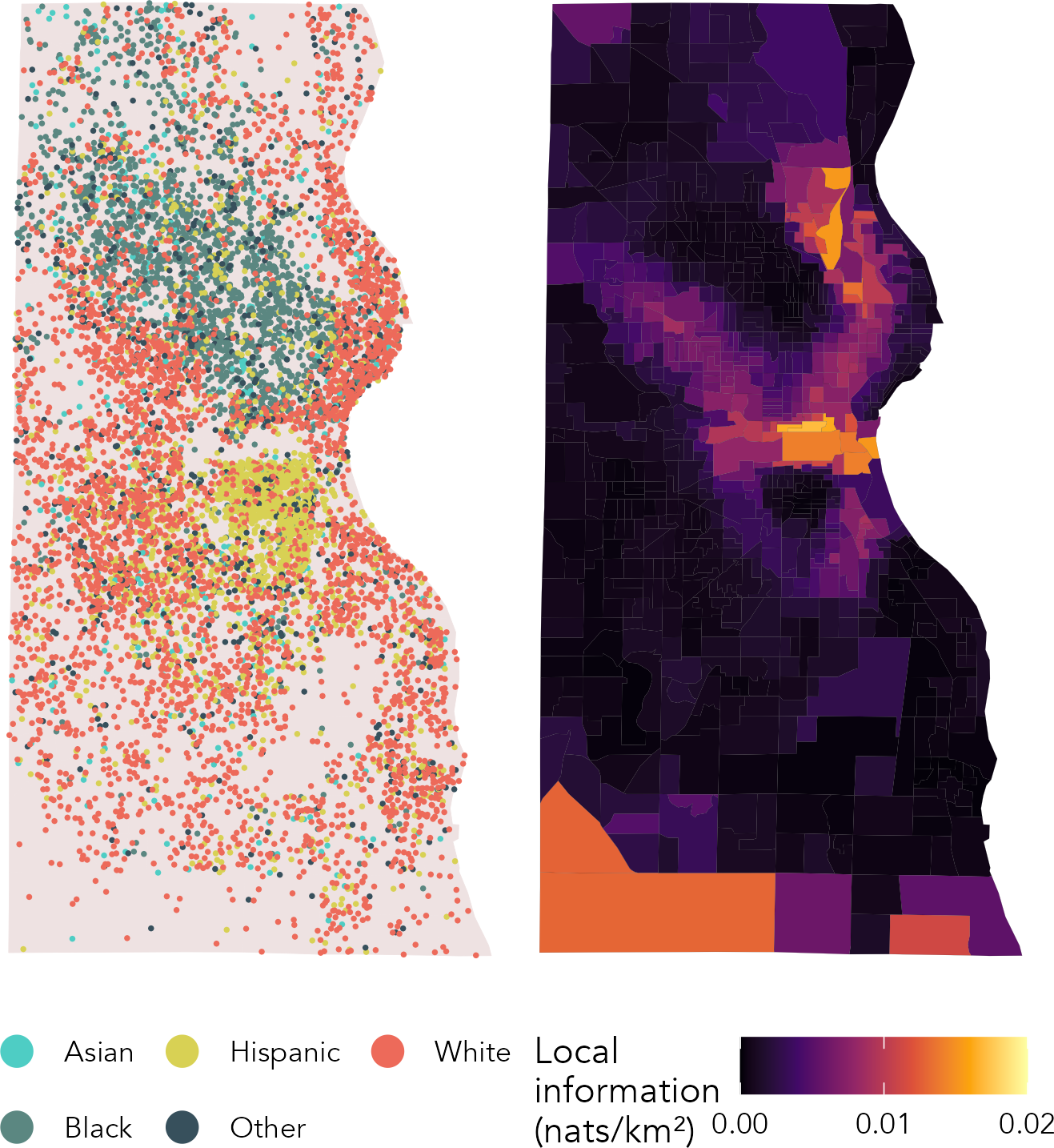}
      \caption{Local information in Milwaukee, Wisconsin. A smoothing kernel of bandwidth 1km was applied to the original population data, and the local information was then computed at each block group using \cref{eq:local-info-trace} with local linear regression to estimate the spatial derivatives of $\Upsilon$.
      The map contains 10,000 dots corresponding to equal units of population within each demographic category and block group.}
      \label{fig:local-info}
    \end{figure}
    
    The proof relies on the local quadratic property of Bregman divergences \eqref{eq:local-quadratic}. 
    The local information \eqref{eq:local-info-trace}  captures the dependence of a spatial measure on \emph{both} the spatial distribution of the population $\Upsilon$ and the choice of diversity function $\phi$.
    To compute the local information in practice, we must approximate the spatial derivatives of $\Upsilon$. 
    I previously performed this task via smoothing followed by local least-squares approximation \cite{chodrowStructureInformationSpatial2017}, but other estimation methods may prove preferable.

    \Cref{fig:local-info} illustrates the local information racial segregation in Milwaukee, Wisconsin, again using the Shannon mutual information $I_{\phi_H}$ as the segregation measure.
    The local information is large along boundaries where the demographic distribution changes rapidly, and is small in the interior of regions of relatively homogeneous demographic composition. 
    In previous work, I found that the mean local information can be used in tandem with the Jensen information as a two-dimensional description of the structure of spatial segregation in a population \cite{chodrowStructureInformationSpatial2017}.

\section{Discussion} \label{sec:discussion}

  One of my primary aims in this article has been to generalize away from particular measures of diversity and segregation and to instead highlight properties and methods that can be obtained for broad classes of measures. 
  To apply the spatial methods described in the previous section, for example, one requires only a choice of strictly convex function $\phi$ that encodes an intuition about a relevant sense of ``diversity.'' 
  To once more highlight the advantage of this general approach, we note that \emph{ordinal} segregation measurement (with respect to variables such as wealth, age, or educational attainment), often treated as an entirely distinct problem in the sociological literature, can be approached using the same toolbox simply with the choice of an appropriate convex $\phi$ such as $\phi_{\mathrm{eq}}$ described in \Cref{example:equality}.
  All properties and analyses described previously are available for this measure except for those which depend on symmetry or separability of $\phi$ (e.g., \Cref{thm:exchanges}).
  I hope that this combination of data-analytic flexibility with a unified theoretical foundation will be useful in future work on segregation measurement. 
  
  A relatively common framing in much of the segregation measurement literature is that one desires to find ``the best'' measure of segregation. 
  I suggest that this framing is somewhat misguided: if one desires to clearly distinguish each of the panels (a-d) of \Cref{fig:checkerboard} from each other, for example, at least two different measurements are needed. 
  One feature of the framework presented here is that one can construct multiple measurements with principled conceptual relationships to each other. 
  Once we formulate a core intuition about what we wish to mean by ``diversity'' in a given setting, we can express this intuition via a strictly convex function $\phi$ and from this derive measurements of diversity, segregation, and spatial segregation. 
  I suggest that we view these as \emph{families} of related measurements, which complement each other in data analysis. 

  \subsection*{Future Directions}

    There are many approaches to the measurement of spatial segregation---and therefore directions of future work---which fall outside of the framework developed here. 
    Even within this framework, however, there are several interesting directions of future work. 
    One is in the design of ordinal diversity functions along the lines of $\phi_{\mathrm{eq}}$ described above; a systematic comparison of alternatives would be of interest. 
    Another direction concerns the measurement of diversity and segregation in the case of more complex joint distributions over multiple demographic categories, such as race and income. 

    Algorithmically, the demographic regionalization problem remains an area of active development. 
    The interpretability of boundaries supplied by the agglomerative hierarchical clustering algorithm of \cite{chodrowStructureInformationSpatial2017} is helpful, but the high-level boundaries are quite sensitive to small perturbations of the input data. 
    Developing alternative hierarchical methods with better properties at the macroscopic level, perhaps via divisive rather than agglomerative clustering, would be a considerable advance. 

    Finally, the \emph{dynamics} of segregation---how the segregation structure of a population changes over time---remains an interesting and challenging area of study. 
    Although several studies (e.g. \cite{hwangRacializedReshufflingUrban2022}) track changes in the value of segregation measures over time, this author is aware of relatively few works connecting mechanistic or agent-based modeling of segregation with empirical data; one notable exception is \cite{benensonSchellingSpatiallyExplicit2009}. 
    Further work towards measuring and modeling the dynamics of spatial segregation would be of considerable interest.

\bibliography{refs}

@article{chodrowEquivalenceInformationsCharacterizes2025b,
  title = {Equivalence of Informations Characterizes {{Bregman}} Divergences},
  author = {Chodrow, Philip S.},
  year = 2025,
  journal = {Entropy},
  volume = {27},
  number = {7},
  pages = {766},
  langid = {english}
}

@article{chodrowStructureInformationSpatial2017,
  title = {Structure and Information in Spatial Segregation},
  author = {Chodrow, Philip S.},
  year = 2017,
  month = oct,
  journal = {Proceedings of the National Academy of Sciences},
  volume = {114},
  number = {44},
  pages = {11591--11596},
  issn = {0027-8424, 1091-6490},
  urldate = {2023-11-01},
  langid = {english}
}

@article{banerjeeClusteringBregmanDivergences2005,
  title = {Clustering with {{Bregman}} Divergences},
  author = {Banerjee, Arindam and Merugu, Srujana and Dhillon, Inderjit and Ghosh, Joydeep},
  year = 2005,
  month = apr,
  journal = {Journal of Machine Learning Research},
  volume = {6},
  pages = {1705--1749},
  urldate = {2025-06-23}
}

@article{reardonMeasuresMultigroupSegregation2002,
  title = {Measures of Multigroup Segregation},
  author = {Reardon, Sean F. and Firebaugh, Glenn},
  year = 2002,
  month = aug,
  journal = {Sociological Methodology},
  volume = {32},
  number = {1},
  pages = {33--67},
  issn = {0081-1750, 1467-9531},
  doi = {10.1111/1467-9531.00110},
  urldate = {2024-07-28},
  copyright = {http://journals.sagepub.com/page/policies/text-and-data-mining-license},
  langid = {english}
}

@article{duchin2023measuring,
  title = {Measuring Segregation via Analysis on Graphs},
  author = {Duchin, Moon and Murphy, James M and Weighill, Thomas},
  year = 2023,
  journal = {SIAM Journal on Matrix Analysis and Applications},
  volume = {44},
  number = {1},
  pages = {80--105},
  publisher = {SIAM}
}

@article{frankelMeasuringSchoolSegregation2011,
  title = {Measuring School Segregation},
  author = {Frankel, David M. and Volij, Oscar},
  year = 2011,
  month = jan,
  journal = {Journal of Economic Theory},
  volume = {146},
  number = {1},
  pages = {1--38},
  issn = {00220531},
  doi = {10.1016/j.jet.2010.10.008},
  urldate = {2024-07-28},
  copyright = {https://www.elsevier.com/tdm/userlicense/1.0/},
  langid = {english}
}

@article{openshaw1984modifiable,
  title = {The Modifiable Areal Unit Problem},
  author = {Openshaw, Stan},
  year = 1984,
  journal = {Concepts and Techniques in Modern Geography},
  publisher = {GeoBooks}
}

@misc{hoogstraQuantitativeMetricsTrait2025,
  title = {Quantitative Metrics for Trait and Identity Distributions},
  author = {Hoogstra, Leah and Slyman, Katherine and Sandstede, Bjorn},
  year = 2025,
  month = aug,
  number = {arXiv:2509.14237},
  eprint = {2509.14237},
  primaryclass = {math.GM},
  publisher = {arXiv},
  doi = {10.48550/arXiv.2509.14237},
  urldate = {2026-06-11},
  archiveprefix = {arXiv}
}

@article{topaz2025diversity,
  title = {Diversity, Identity, and Data},
  author = {Topaz, Chad M and Brooks, Heather Z and Kan, Unchitta and Sandstede, Bj{\"o}rn and Smith, Christian Michael},
  year = 2025,
  journal = {The American Mathematical Monthly},
  volume = {132},
  number = {1},
  pages = {36--47},
  publisher = {Taylor \& Francis}
}

@book{pielou1977mathematical,
  title={Mathematical Ecology},
  author={Pielou, Evelyn C},
  volume=385,
  year=1977,
  publisher={Wiley New York}
}

@article{duncanMethodologicalAnalysisSegregation1955,
  title = {A Methodological Analysis of Segregation Indexes},
  author = {Duncan, Otis Dudley and Duncan, Beverly},
  year = 1955,
  month = apr,
  journal = {American Sociological Review},
  volume = {20},
  number = {2},
  eprint = {2088328},
  eprinttype = {jstor},
  pages = {210},
  issn = {00031224},
  doi = {10.2307/2088328},
  urldate = {2025-01-08}
}

@article{chakrabarti2005shannon,
  title={Shannon entropy: axiomatic characterization and application},
  author={Chakrabarti, CG and Chakrabarty, Indranil},
  journal={International Journal of Mathematics and Mathematical Sciences},
  volume={2005},
  number={17},
  pages={2847--2854},
  year={2005},
  publisher={Wiley Online Library}
}

@article{reardon3MeasuresSpatial2004,
  title = {Measures of Spatial Segregation},
  author = {Reardon, Sean F. and O'Sullivan, David},
  year = 2004,
  month = dec,
  journal = {Sociological Methodology},
  volume = {34},
  number = {1},
  pages = {121--162},
  issn = {0081-1750, 1467-9531},
  doi = {10.1111/j.0081-1750.2004.00150.x},
  urldate = {2025-01-09},
  copyright = {https://journals.sagepub.com/page/policies/text-and-data-mining-license},
  langid = {english}
}

@article{robertoSpatialProximityConnectivity2018,
  title = {The Spatial Proximity and Connectivity Method for Measuring and Analyzing Residential Segregation},
  author = {Roberto, Elizabeth},
  year = 2018,
  month = aug,
  journal = {Sociological Methodology},
  volume = {48},
  number = {1},
  pages = {182--224},
  issn = {0081-1750, 1467-9531},
  doi = {10.1177/0081175018796871},
  urldate = {2024-10-22},
  langid = {english}
}

@article{kirkleySpatialRegionalizationBased2022,
  title = {Spatial Regionalization Based on Optimal Information Compression},
  author = {Kirkley, Alec},
  year = 2022,
  month = oct,
  journal = {Communications Physics},
  volume = {5},
  number = {1},
  pages = {249},
  issn = {2399-3650},
  doi = {10.1038/s42005-022-01029-4},
  urldate = {2026-08-14},
  langid = {english}
}

@book{boydConvexOptimization2023,
  title = {Convex Optimization},
  author = {Boyd, Stephen P. and Vandenberghe, Lieven},
  year = 2023,
  edition = {Version 29},
  publisher = {Cambridge University Press},
  address = {Cambridge, UK},
  isbn = {978-0-521-83378-3},
  langid = {english}
}

@incollection{reardonMeasuresOrdinalSegregation2009,
  title = {Measures of Ordinal Segregation},
  booktitle = {Research on {{Economic Inequality}}},
  author = {Reardon, Sean F.},
  editor = {Fl{\"u}ckiger, Yves and Reardon, Sean F. and Silber, Jacques},
  year = 2009,
  month = jan,
  volume = {17},
  pages = {129--155},
  publisher = {Emerald Group Publishing Limited},
  doi = {10.1108/S1049-2585(2009)0000017011},
  urldate = {2024-07-28},
  copyright = {http://www.emeraldinsight.com/page/tdm},
  isbn = {978-1-84855-786-4 978-1-84855-787-1},
  langid = {english}
}

@article{hwangRacializedReshufflingUrban2022,
  title = {Racialized Reshuffling: Urban Change and the Persistence of Segregation in the Twenty-First Century},
  shorttitle = {Racialized {{Reshuffling}}},
  author = {Hwang, Jackelyn and McDaniel, Tyler W.},
  year = 2022,
  month = jul,
  journal = {Annual Review of Sociology},
  volume = {48},
  number = {1},
  pages = {397--419},
  issn = {0360-0572, 1545-2115},
  doi = {10.1146/annurev-soc-030420-014126},
  urldate = {2026-09-09},
  langid = {english}
}

@article{benensonSchellingSpatiallyExplicit2009,
  title = {From Schelling to Spatially Explicit Modeling of Urban Ethnic and Economic Residential Dynamics},
  author = {Benenson, Itzhak and Hatna, Erez and Or, Ehud},
  year = 2009,
  month = may,
  journal = {Sociological Methods \& Research},
  volume = {37},
  number = {4},
  pages = {463--497},
  issn = {0049-1241, 1552-8294},
  doi = {10.1177/0049124109334792},
  urldate = {2026-09-09},
  copyright = {https://journals.sagepub.com/page/policies/text-and-data-mining-license},
  langid = {english}
}

@book{cover2012elements,
  title = {Elements of Information Theory},
  author = {Cover, Thomas M and Thomas, Joy A},
  year = 2012,
  publisher = {John Wiley \& Sons}
}

\appendix

\section{Pielou's Third Diversity Axiom} \label{sec:pielou-third-axiom}

The last of Pielou's three axioms relates to the behavior of a diversity measure when multiple, independent demographic attributes are considered simultaneously.

\textbf{Additivity under independent attributes.} 
``\emph{In a multiple classification (two or more qualitative variates) where the classifications are independent, the total diversity should equal the sum of the respective diversities.}'' {\cite{pielou1977mathematical}}.  
This statement can be formalized as follows, interpreting ``independent'' in the sense of probability distributions: suppose that, for a given population $\pop = \standardTriple \in \Pops_{nd\ell}$, there exist two other populations $\pop_1 = (\mX, \mU \in \R^{n\times \ell_1}, \mu)$ and $\pop_2 = (\mX, \mV \in \R^{n\times \ell_2}, \mu)$ such that $\ell = \ell_1\ell_2$ and a bijection $\alpha: [\ell_1] \times [\ell_2] \rightarrow [\ell]$ such that $y_{ij} = u_{i\alpha_1(j)} v_{i\alpha_2(j)}$ for all $i \in [n]$ and $j \in [\ell]$. 
Then $f(\pop) = f(\pop_1) + f(\pop_2)$.

This axiom, in combination with Pielou's other two and under the mild assumption of separability and continuity, is sufficient to uniquely characterize the Shannon entropy: 
\begin{thm}
    If a function $\psi: \simplex{\ell} \rightarrow \R$ is strictly concave, continuous\footnote{Convex functions are continuous on the interior of their domain, so this stipulation requires only that they are also continuous at the boundary.}, separable,  and satisfies the additivity under independent attributes property, then  
    \begin{align}
      \psi(\vy) = -c\sum_{j = 1}^{\ell} y_j \log y_j \; \label{eq:shannon-entropy}
    \end{align}
    for some constant $c > 0$.
\end{thm}
A proof is given by Chakrabarti and Chakrabarty {\cite{chakrabarti2005shannon}}, where additivity under independent attributes corresponds to the authors' Axiom 4.\footnote{The authors' statement of their result also hypothesizes nonnegativity of $\phi$, but this hypothesis is not required by their proof.}

\section{Supplementary Proofs} \label{sec:supplementary-proofs}

\begin{proof}[Proof of \Cref{thm:monotone-under-group-addition}]
    Let $\phi$ be strictly convex. 
    We'll first assume that $\phi$ is nonnegative. 
    Then, for any $I \subset [\ell]$ and $\lambda = \frac{|I|}{\ell}$, we have
    \begin{align*}
      \phi\paren{\frac{1}{\ell} \ve_{[\ell]} } 
      & = \phi\paren{\lambda \frac{1}{|I|} \ve_{I} + (1-\lambda) \frac{1}{\ell - |I|} \ve_{[\ell]\setminus I}} \\
      & < \lambda \phi\paren{\frac{1}{|I|} \ve_{I}} + (1-\lambda) \phi\paren{\frac{1}{\ell - |I|} \ve_{[\ell]\setminus I}} \\
      & \leq \phi\paren{\frac{1}{|I|} \ve_{I}} \;,
    \end{align*}
    where the first inequality is by convexity and the second is by nonnegativity. 

    Now, relax the requirement that $\phi$ is nonnegative. 
    Since the simplex $\simplex{\ell}$ is compact (closed and bounded) and $\phi$ is continuous, $\phi$ attains a minimum on $\simplex{\ell}$, which we denote $\phi_{\min} = \min_{\vy \in \simplex{\ell}} \phi(\vy)$.
    Then, let $\phi'(\vy) = \phi(\vy) - \phi_{\min}$, which is nonnegative and convex. 
    The calculation above now goes through replacing $\phi$ with $\phi'$, proving that $\phi'\paren{\frac{1}{\ell} \ve_{[\ell]}} \leq \phi'\paren{\frac{1}{|I|} \ve_{I}}$  and therefore that $\phi\paren{\frac{1}{\ell} \ve_{[\ell]}} \leq \phi\paren{\frac{1}{|I|} \ve_{I}}$.
  \end{proof}

\begin{proof}[Proof of \Cref{thm:diversity-axioms-1}]
      To demonstrate the uniform maximizer property, for any vector $\vy \in \simplex{\ell}$,
      \begin{align*}
        \ve_{[\ell]}/\ell = \frac{1}{\abs{\cS}}\sum_{\sigma \in \cS} \sigma(\vy) \;,
      \end{align*}
      where $\cS$ is the set of all permutations of $[\ell]$.
      It follows that 
      \begin{align*}
        \psi\paren{\frac{1}{\ell} \ve_{[\ell]}} 
        & = \psi\paren{\frac{1}{\abs{\cS}}\sum_{\sigma \in \cS} \sigma(\vy)} \\
        & < \frac{1}{\abs{\cS}}\sum_{\sigma \in \cS} \psi\paren{\sigma(\vy)} = \psi(\vy) \;,
      \end{align*}
      where the inequality is by an inductive application of strict concavity and the final equality is by symmetry.
  \end{proof}

\begin{proof}[Proof of \Cref{thm:exchanges}]
  The proof is by direct calculation. Let $\pop' = \exc{a}{b}{i}{j}{\epsilon}(\pop)$, and let $\vy_i'$ and $\vy_j'$ be the demographic vectors at locations $i$ and $j$ in $\pop'$. 
  Since the global demographic vector $\bar{\vy}$ is unchanged by the exchange, we have
  \begin{align*}
    I_{\phi}(\pop) - I_{\phi}(\pop') &= \sum_{i = 1}^n \brackets{\mu_i \phi(\vy_i) - \mu_i' \phi(\vy_i')} \\ 
    &= \mu_i\brackets{ \phi(\vy_i) - \phi(\vy_i')} + \mu_j\brackets{ \phi(\vy_j)  - \phi(\vy_j')} \\ 
    &= \mu_i \nabla \phi(\vy_i)^{\top} (\vy_i - \vy_i') \\ 
    & + \mu_j \nabla \phi(\vy_j)^{\top} (\vy_j - \vy_j') \\ 
    &+ o(\norm{\vy_i - \vy_i'}) + o(\norm{\vy_j - \vy_j'})\;. 
  \end{align*}
  since $\vy_k = \vy_k'$ for all $k \notin \{i, j\}$ and $\mu_k = \mu_k'$ for all $k$.
  Using now that $\vy_i - \vy_i' = \lambda_i (\tilde{\ve}_a - \tilde{\ve}_b)$ and $\vy_j - \vy_j' = \lambda_j (\tilde{\ve}_b - \tilde{\ve}_a)$ and ignoring the higher-order terms, we have 
  \begin{align*}
    I_{\phi}(\pop) - I_{\phi}(\pop') &\simeq \mu_i \lambda_i \nabla \phi(\vy_i)^{\top} (\tilde{\ve}_a - \tilde{\ve}_b) + \mu_j \lambda_j \nabla \phi(\vy_j)^{\top} (\tilde{\ve}_b - \tilde{\ve}_a)\;. 
  \end{align*}
  By the separability assumption, we have $\nabla \phi(\vy) = (\phi'(y_1), \ldots, \phi'(y_{\ell}))^{\top}$, and so 
  \begin{align*}
    I_{\phi}(\pop) - I_{\phi}(\pop') &\simeq \mu_i \lambda_i (\phi'(y_{ia}) - \phi'(y_{ib})) + \mu_j \lambda_j (\phi'(y_{jb}) - \phi'(y_{ja})) \\ 
    &= \epsilon (\phi'(y_{ia}) - \phi'(y_{ib}) + \phi'(y_{jb}) - \phi'(y_{ja})) \\ 
    &= \epsilon (\phi'(y_{ia}) - \phi'(y_{ja}) + \phi'(y_{jb}) - \phi'(y_{ib})) \\ 
    &> 0\;.
  \end{align*}
  The last line follows from the assumptions that $y_{ia} < y_{ja}$ and $y_{ib} > y_{jb}$, and the fact that, if $\phi$ is strictly convex, then $\phi'$ is strictly increasing, and so $\phi'(y_{ia}) - \phi'(y_{ja}) > 0$ and $\phi'(y_{jb}) - \phi'(y_{ib}) > 0$.
  Choosing $\epsilon$ sufficiently small to justify the linear approximation completes the proof. 
\end{proof}

\begin{proof}[Derivation of \cref{eq:wealth-jensen-info}]
  We have 
  \begin{align*}
    I_{\phi_{\mathrm{eq}}}(\pop) &= \sum_{i = 1}^n \mu_i \paren{\norm{\mC(\vy_i - \vyref)}^2 - \norm{\mC(\bar{\vy} - \vyref)}^2} \\
    &= \sum_{i = 1}^n \mu_i \paren{\norm{\mC\vy_i}^2 - 2\inner{\mC\vy_i}{\mC\vyref} + \norm{\mC\vyref}^2} \\ 
    &\quad - \sum_{i = 1}^n \mu_i\paren{\norm{\mC\bar{\vy}}^2 + 2\inner{\mC\bar{\vy}}{\mC\vyref} - \norm{\mC\vyref}^2} \\
    &= \sum_{i = 1}^n \mu_i \paren{\norm{\mC\vy_i}^2 - 2\inner{\mC\vy_i}{\mC\vyref} - \norm{\mC\bar{\vy}}^2 + 2\inner{\mC\bar{\vy}}{\mC\vyref}} \\
    &= \sum_{i = 1}^n \mu_i \paren{\norm{\mC\vy_i}^2 - 2\inner{\mC(\vy_i - \bar{\vy})}{\mC\vyref} - \norm{\mC\bar{\vy}}^2 }\\
    &= \paren{\sum_{i = 1}^n \mu_i \norm{\mC\vy_i}^2} - \norm{\mC\bar{\vy}}^2 \;,
  \end{align*}
  since $\sum_{i = 1}^n \mu_i \vy_i = \bar{\vy}$. 
\end{proof}

\end{document}